\documentclass[12pt,draftclsnofoot,journal,onecolumn]{IEEEtran}
\usepackage{mathtools}
\usepackage{ifpdf}
\usepackage{amsthm}
\usepackage{amsmath}
\usepackage{nicefrac}
\usepackage{dsfont}
\usepackage{cite}
\usepackage{amsfonts}
\usepackage{comment}
\usepackage{bbm}
 \usepackage{url}
 \usepackage{amssymb}
 \usepackage[paper=letterpaper,margin=0.75in]{geometry}
 \usepackage[ansinew]{inputenc}

\usepackage{lipsum}
\usepackage{hyperref}
\theoremstyle{definition}
\newtheorem{thm}{Theorem}

\begin{document}
\title{Fundamental Limits of Covert Bit Insertion in Packets}

\author{ \IEEEauthorblockN{Ramin Soltani\IEEEauthorrefmark{1},
		Dennis Goeckel\IEEEauthorrefmark{1}, Don Towsley\IEEEauthorrefmark{2}, and Amir Houmansadr\IEEEauthorrefmark{2}}
	
	\IEEEauthorblockA{\IEEEauthorrefmark{1}Electrical~and~Computer~Engineering~Department,~University~of~Massachusetts,~Amherst,
		\{soltani, goeckel\}@ecs.umass.edu\\}
	\IEEEauthorblockA{\IEEEauthorrefmark{2}College of Information and Computer Sciences, University of Massachusetts, Amherst,
		\{towsley, amir\}@cs.umass.edu}

                       \thanks{This work has been supported by the National Science Foundation under grants CNS-1564067 and CNS-1525642.}
                         \thanks{ This work has been presented at the 56th Annual Allerton Conference on Communication, Control, and Computing,  October 2018.}
                       \thanks{Personal use of this material is permitted. Permission from IEEE must be obtained for all other uses, in any current or future media, including reprinting/republishing this material for advertising or promotional purposes, creating new collective works, for resale or redistribution to servers or lists, or reuse of any copyrighted component of this work in other works.}
}

\date{}
\maketitle
\thispagestyle{plain}
\pagestyle{plain}

\begin{abstract}
Covert communication is necessary when revealing the mere existence of a message leaks sensitive information to an attacker. Consider a network link where an authorized transmitter Jack sends packets to an authorized receiver Steve, and the packets visit Alice, Willie, and Bob, respectively, before they reach Steve. Covert transmitter Alice wishes to alter the packet stream in some way to send information to covert receiver Bob without watchful and capable adversary Willie being able to detect the presence of the message. In our previous works, we addressed two techniques for such covert transmission from Alice to Bob: packet insertion and packet timing. In this paper, we consider covert communication via bit insertion in packets with available space (e.g., with size less than the maximum transmission unit). We consider three scenarios: 1) packet sizes are independent and identically distributed (i.i.d.) with a probability mass function (pmf) whose support is a set of one bit spaced values; 2) packet sizes are i.i.d. with a pmf whose support is arbitrary; 3) packet sizes may be dependent. For the first and second assumptions, we show that Alice can covertly insert $\mathcal{O}(\sqrt{n})$ bits of information in a flow of $n$ packets; conversely, if she inserts $\omega(\sqrt{n})$ bits of information, Willie can detect her with arbitrarily small error probability. For the third assumption, we prove Alice can covertly insert on average $\mathcal{O}(c(n)/\sqrt{n})$ bits in a sequence of $n$ packets, where $c(n)$ is the average number of conditional pmf of packet sizes given the history, with a support of at least size two.
\end{abstract}

\textbf{Keywords: Covert Communication, Bit Insertion, Security and Privacy, Computer Networks, UDP, TCP/IP} 

%%%%%%%%%%%%%%%%%%%%%%%%%%%%%%%%%%%%%%%%%%%%%%%%%%%%%%%%%%%%%%%%%%%%%%%%%%%%%%%%%%%%%%%%%%%%%%%%%%%%%%%%%%%

\begin{comment}

\author{
\IEEEauthorblockN{Ramin Soltani\IEEEauthorrefmark{1},
Dennis Goeckel\IEEEauthorrefmark{1}, Don Towsley\IEEEauthorrefmark{2}, and Amir Houmansadr\IEEEauthorrefmark{2}}

\IEEEauthorblockA{\IEEEauthorrefmark{1}Electrical and Computer Engineering Department, University of Massachusetts, Amherst,
\{soltani, goeckel\}@ecs.umass.edu}
\IEEEauthorblockA{\IEEEauthorrefmark{2}College of Information and Computer Sciences, University of Massachusetts, Amherst,
\{towsley, amir\}@cs.umass.edu}

\thanks{This work has been supported by the National Science Foundation under grants CNS-1564067 and CNS-1525642.}
}

\end{comment}

\section{Introduction}
\IEEEPARstart{W}{ith} the rapid growth of communication systems and the Internet, security and privacy have emerged as critical concerns~\cite{nichols2001wireless,takbiri2018asymptotic,ISIT17-longversion1,hadian2016privacy,hadian2018privacy,nazanin_ISIT2018}. Protecting the content of messages through encryption~\cite{talb2006} or information-theoretic methods~\cite{bloch11pls} forms the majority of the research on security. However, in many scenarios, the security is achieved by hiding not only the message but also the existence of the message. In other words, if the existence of the message is revealed, it can leak sensitive information to the adversary~\cite{snowden} or cause threats to users. Such scenarios include a situation where people do not have the freedom to communicate, military applications, and protecting location privacy of mobile users. Covert communication provides the solution in such scenarios by hiding the existence of the communication.

Although spread spectrum~\cite{simon94ssh} and steganography~\cite{ker07pool, ker2009square,filler2009square} have been studied broadly~\cite{petitcolas1999information}, only recently have the fundamental limits of covert communication over noisy continuous-valued channels been studied.  Bash et al.~\cite{bash_isit2012,bash_jsac2013} showed that the number of bits that can be transmitted covertly and reliably on an additive white Gaussian noise (AWGN) channel is on the order of the square root of the number of channel uses. The work of Bash et al. motivated a significant body of work~\cite{jaggi_isit2013,bash_isit2013,bash_isit2014,bash2015hiding,bloch2016covert,soltani2014covert,soltani2015covert,soltani2016allerton,soltani2017towards,soltani2018covert,soltani2018fundamental,tahmasbi2017first,arumugam2018embedding}, much of which is focused on models appropriate for point-to-point wireless communication channels.

\begin{figure}
\begin{center}
\includegraphics[width=\textwidth/2,height=\textheight,keepaspectratio]{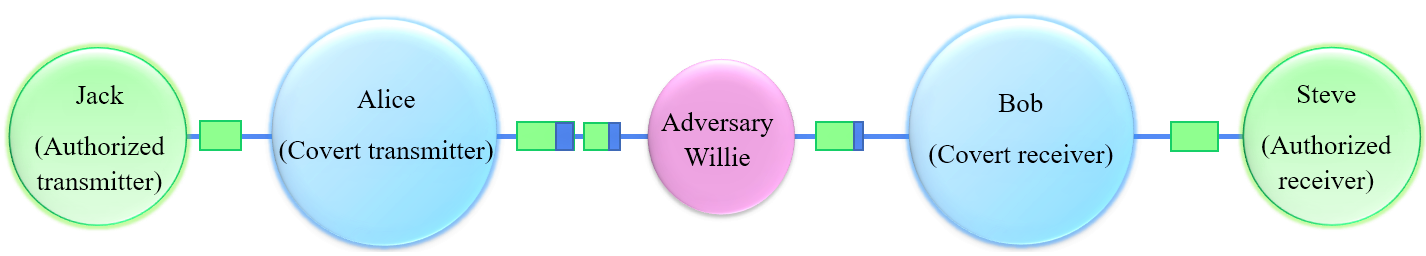}
\end{center}
 \caption{System configuration: Authorized user Jack sends packets through the network to authorized user Steve. Alice adds her information to packets transmitted by Jack to communicate covertly with Bob without detection by the (adversary) warden Willie. Bob removes Alice's inserted bits from the packets. In each packet, the blue part shows the information added by Alice and the green part shows Jack's information.}
 \label{fig:1}
 \end{figure}
 
This paper focuses on the covert transmission of information on a network link. Consider a link where an authorized user (Jack) transmits a packet stream of length $n$ to an authorized receiver (Steve). The link is watched by an authority, warden Willie, who seeks to detect whether anybody other than the authorized users Jack and Steve might be using the link.  Indeed, a covert transmitter Alice might attempt to transmit information reliably and covertly to a covert receiver Bob. After originating at Jack, the packet stream visits Alice, Willie, Bob, and Steve, respectively, as shown in Fig~\ref{fig:1}. In~\cite{soltani2015covert,soltani2016allerton}, we presented two techniques for covert communication between Alice and Bob for the case that Jack's packet timings are governed by a Poisson process ~\cite{soltani2015covert} or a renewal process~\cite{soltani2016allerton}. If Willie can verify the sources of the packets, then Alice can transmit $\mathcal{O}(n)$ bits to Bob via timing if the timing channel between Alice and Bob has non-zero capacity. If the timing channel between Alice and Bob has zero capacity and Willie cannot authenticate packet senders, Alice can transmit $\mathcal{O}(\sqrt{n})$ bits to Bob via packet insertion.

In this paper, we turn our attention to a third technique for covert communication between Alice and Bob: employing the space in the payload of the packets. In~\cite{handel1996hiding}, Handel et al. propose information embedding in the two least significant bits of packet headers as well as the time stamp. Rowland et al.~\cite{rowland1997covert} proposed information embedding in the TCP/IP header fields such as the IP identification field, the initial sequence number field, and the TCP acknowledge sequence number field. Furthermore, there is a significant body of work related to embedding information in the header of the packets~\cite{fisk2002eliminating,ahsan2002practical,shah2011network,cauich2005data}. %However, Murdoch et al.~\cite{murdoch2005embedding} propose that TCP/IP does not allow easy information embedding in the least significant bit of the timestamp or the IP identifier of the header because these fields are structured and non-uniform.
Here, we do not alter the original information in the payload of the packets, insert new packets, or employ timing channels. Rather, we consider adding extra information to the payload of the packets when packet sizes are variable without an adversary noticing a change in the packet length distribution. To the best of our knowledge, no previous work offers theoretical guarantees on the trade-off between the efficiency of bit insertion in payloads (number of inserted bits) and covertness. 

We consider a setting similar to those of~\cite{soltani2015covert,soltani2016allerton}. Consider a network link where an authorized transmitter (Jack) sends a flow of $n$ packets to an authorized receiver (Steve). Alice wishes to employ the available (unused) space in the packets to insert her information to communicate covertly with Bob, in the presence of the network warden, Willie, who is observing the sizes and timings of the packets transmitted by Alice to detect such transmissions (see Fig.~\ref{fig:1}). Alice can use one bit of the header (as motivated by~\cite{fisk2002eliminating,ahsan2002practical,shah2011network,cauich2005data,rowland1997covert,handel1996hiding}), which we refer to as ``flag bit'', and the available space in the packets. Willie cannot modify packets. He knows the joint probability mass function (pmf) of the packet sizes of the stream transmitted by Jack, so he seeks to apply hypothesis testing to verify whether the packet process has the proper characteristics. First we consider Assumption 1, i.e., packet sizes are i.i.d. with a pmf whose support is a set of one bit spaced values across an interval $[S_{\min},S_{\max}]$, and we show that Alice can insert $\mathcal{O}\left(\sqrt{n}\right)$ bits in the packet stream to communicate with Bob while lower bounding Willie's error probability $\mathbb{P}_{e}$ by $\frac{1}{2}-\epsilon$ for any $\epsilon>0$ (see Theorem~\ref{thm:iid}). Conversely, we prove that if Alice inserts ${\omega}\left(\sqrt{n}\right)$ bits in the sequence of $n$ packets, Willie will detect her with arbitrarily small error probability.

To establish the result, we employ the following construction. Alice generates a secret key of length $\mathcal{O}(\sqrt{n}\log{n})$ and shares it with Bob before the communication. The key indicates the location of the $\mathcal{O}(\sqrt{n})$ packets to be selected by Alice for a possible bit insertion. If the selected packet has available space, i.e., its size satisfies $S \leq S_{\max}-\lfloor K/2 \rfloor$, where $K$ is the number of possible sizes implied by $f(x)$, then Alice adds $\lfloor K/2 \rfloor$ bits to the packet and changes the flag bit to one; otherwise, she only changes the flag bit to zero. On the other side, Bob extracts Alice's information using the shared key and the flag bits. We extend our results to the case where the possible packet sizes are spaced arbitrarily (Assumption 2), and the case where packet sizes may be dependent (Assumption 3).

The remainder of the paper is organized as follows. We present the system model and the metrics in Section~\ref{modcons}. Then, we present construction and analysis for Assumptions~1,~2, and~3 in Sections~\ref{sec:iid},~\ref{sec:iid2}, and~\ref{sec:niid}, respectively. In Section~\ref{sec:f} we present the future work. Finally, we discuss the assumptions and results in Section~\ref{dis}, and we conclude in Section~\ref{con}.

\section{System Model and Metrics}
\label{modcons}

\subsection{System Model}
Suppose that Jack transmits a sequence of $n$ packets to Steve on a network link, and Alice wishes to add her information to the packets on the link to communicate with Bob without being detected by adversary Willie. Alice is allowed to use a bit from the header (e.g., least significant bit) which we refer to as a ``flag bit'' and the available space in the payload of each packet. She generates a secret key of length $\sqrt{n}\log{n}$ bits and shares it with Bob before insertion. The key indicates the packets that will possibly be manipulated by Alice. Willie knows the joint pmf of the packet sizes and Alice's bit insertion scheme but not the key. He cannot observe the content of the packets, but he observes the sizes of the packets coming from Alice and attempts to detect Alice's transmission (see Fig.~\ref{fig:1}). Willie cannot observe the contents of the packets or modify them. We consider the following assumptions:
\begin{itemize}
	\item Assumption 1: The packet sizes are independent and identically distributed (i.i.d.) with pmf $f(x)$ whose support is a set of one bit spaced values across an interval $[S_{\min},S_{\max}]$, and the number of packets sizes is at least two, i.e.,  
	\begin{align}
	K=\nonumber |{\displaystyle \{x\in \mathbb {R} :f(x)>0\}}| \geq 2.
	\end{align}
	\noindent where $|\cdot|$ denotes cardinality of a set. 
	\item Assumption 2: The packet sizes are i.i.d. and with pmf $f(x)$ with an arbitrary support, and the number of packet sizes is at least two.
	\item Assumption 3: The packet sizes might be dependent with joint pmf $f(x_1,\ldots,x_n)$.
\end{itemize}
For the case of i.i.d. packet sizes (Assumptions~1 and~2), we denote the mean and the variance of packet sizes by $\mu$ and $\sigma^2$, respectively, where $\sigma^2<\infty$. We assume that Alice, Willie, and Bob know the pmf of the packet sizes.

\subsection{Hypothesis Testing}
Willie is faced with two hypotheses:
\begin{itemize}
\item $H_0$: Alice does not insert her own bits (null hypothesis)
\item $H_1$: Alice inserts her own bits (alternative hypothesis)
\end{itemize}
He applies a binary hypothesis test to decide between $H_0$ and $H_1$. Denote by $\mathbb{P}_{0}$ and $\mathbb{P}_{1}$ the distributions that Willie observes under $H_0$ and $H_1$, respectively. Willie's detection is associated with two errors:
\begin{itemize}
	\item $\mathbb{P}_{FA}$: probability of detecting $H_1$ when $H_0$ is true (type I error or false alarm) 
	\item $\mathbb{P}_{MD}$: probability of detecting $H_0$ when $H_1$ is true (type II error or missed detection)
\end{itemize}
We assume $\mathbb{P}(H_0)=\mathbb{P}(H_1)$ and that Willie seeks to minimize his probability of error 
$$\mathbb{P}_{e}={\mathbb{P}_{FA} \mathbb{P}(H_0)+ \mathbb{P}_{MD}}\mathbb{P}(H_1)=\frac{\mathbb{P}_{FA} + \mathbb{P}_{MD}}{2}.$$
\noindent Our results are readily extended to the case where $\mathbb{P}(H_0)\neq \mathbb{P}(H_1)$ since minimizing ${\mathbb{P}_{FA} + \mathbb{P}_{MD}}$ is applicable when $\mathbb{P}(H_0)\neq \mathbb{P}(H_1)$~\cite{soltani2018fundamental}.  

\subsection{Covertness}
Alice's transmission is {\em covert} if and only if she can achieve $\mathbb{P}_{e}\geq \frac{1}{2}-\epsilon$ for any $\epsilon>0$~\cite{bash_jsac2013}, which means Willie's detector operates as close as desired to a random detector.
In this paper, we use standard Big-O, little-omega, and Big-Theta notations~\cite[Ch. 3]{cormen2009introduction}.

\section{Assumption 1: Packet sizes are i.i.d. with a pmf whose support is a set of one bit spaced values} \label{sec:iid}

In this section, we assume that packet sizes are i.i.d. with pmf $f(x)$ whose support is a set of one bit spaced values across an interval $[S_{\min}, S_{\max}]$, and it allows at least two packet sizes. We determine the total number of bits that Alice can insert in the packets.   

\begin{thm} \label{thm:iid} 
If packet sizes are i.i.d. with a pmf whose support is a set of one bit spaced values, and there are at least two possible packet sizes, Alice can covertly insert $\mathcal{O}(\sqrt{n})$ bits in a sequence of $n$ packets transmitted by Jack. Conversely, if Alice inserts $\omega\left(\sqrt{n}\right)$ bits in a sequence of $n$ packets, Willie can detect her with arbitrarily small error probability $\mathbb{P}_{e}$.
\end{thm}

\begin{proof}

{\it (Achievability)} Denote by $K=S_{\max}-S_{\min}+1$ the size of the support of $f(x)$, i.e., the number of possible packet sizes. The construction and analysis for the case of odd $K$ follows from those of even $K$ with minor modifications. In particular, when $K$ is odd, Alice and Bob disregard all packets of a specific size (e.g., smaller packet size). We consider even $K$ in this proof.  

\textbf{Construction}: Let $m= K/2 $. Alice generates a secret key of length $\mathcal{O} (\sqrt{n} \log{n})$, to which Willie does not have access. In particular, she first generates an i.i.d. bit sequence of length $n$ in which each bit is one with probability 
\begin{align}
\label{eq:p} p=\frac{\epsilon}{\xi \sqrt{n}},
\end{align}
where
\begin{align}
\label{eq:xi} \xi&=\max\{1,\xi_0\},\\
\label{eq:xi0} \xi_0&=\max\limits_{x\in [S_{\max}-m+1,S_{\max}]}{\frac{f(x-m)}{f(x)}},
\end{align}
The locations of the ones indicate the packets that will be selected by Alice for a possible bit insertion. Hence, the key contains $\epsilon \sqrt{n}$ addresses of maximum length $\lceil \log n \rceil$ for a possible bit insertion.

Alice shares the key with Bob. The key indicates which packets in the stream of length $n$ packets from Jack have been selected by Alice for a possible bit insertion. If the size of the selected packet satisfies $S\leq S_{\max}-m$, then Alice inserts $m$ bits to the end of the payload of the packet and sets the flag bit of the packet to one; otherwise, she only changes the flag bit to zero indicating that the packet did not have available space. 

Bob, who has access to the key and the flag bits finds out which packets contain Alice's bits. Therefore, he extracts and removes the last $m$ bits of those packets.  

\textbf{Analysis}: {\em (Covertness)} If Willie applies hypothesis test, then \cite{bash_jsac2013}
\begin{align} \label{eq:5} \mathbb{P}_{e} \geq \frac{1}{2}- \sqrt{\frac{1}{2} \mathcal{D}(\mathbb{P}_1 || \mathbb{P}_0)},   
\end{align}
\noindent where $\mathbb{P}_1$ and $\mathbb{P}_0$ are the joint pmf of the packet sizes when $H_0$ and $H_1$ are true, respectively, and $\mathcal{D}(\mathbb{P}_1 || \mathbb{P}_0)$ is the relative entropy between $\mathbb{P}_1$ and $\mathbb{P}_0$. Next, we show that Alice can upper bound $\mathcal{D}(\mathbb{P}_1 || \mathbb{P}_0)$ by $2 \epsilon^2$, and thus lower bound Willies error probability, $\mathbb{P}_{e}$, by $1-\epsilon$, for arbitrary $\epsilon$. Recall that packet sizes are i.i.d. with pmf $f(x)$. Since the locations of the ones in the bit sequence are uniformly distributed, and thus Alice selects each packet for insertion independently, the distributions of the packet sizes remains i.i.d. after Alice's manipulation. Denote by $\widetilde{f}(x)$ the pmf of the packet sizes after Alice inserts information in them. Observe:
\begin{align}
\mathbb{P}_0&= \prod_{i=1}^{n}f(x_i),\\
\mathbb{P}_1&= \prod_{i=1}^{n}\widetilde{f}(x_i).
\end{align}
\noindent From the chain rule for relative entropy\cite[Eq.(2.67)]{cover2012elements}:
\begin{equation}\label{eq:6}
\mathcal{D}(\mathbb{P}_1 || \mathbb{P}_0) = n \mathcal{D}(\widetilde{f}(x) || f(x)).
\end{equation}
Next, we calculate $\widetilde{f}(x)$. Note that from Willie's perspective, the probability that a packet is selected by Alice for a possible insertion is $p$. Observe: 
\small
\begin{align}
\label{eq:fofx}   \widetilde{f}(x)= 
 \begin{cases} 
f(x)(1-p), & S_{\min}\leq x\leq S_{\min}+m-1 \\
 f(x)+pf(x-m). & S_{\min}+m\leq x\leq S_{\max}
\end{cases}
\end{align}
\normalsize
\noindent Therefore, 
\begin{align}
\nonumber \mathcal{D}(  \widetilde{f}(x)|| f(x))&= \sum_{x=S_{\min}}^{S_{\min}+m-1} \widetilde{f}(x) \ln \frac{\widetilde{f}(x)}{f(x)},\\
\label{eq:d1} &\phantom{=}+ \sum_{x=S_{\min}+m}^{S_{\max}} \widetilde{f}(x) \ln \frac{\widetilde{f}(x)}{f(x)},
\end{align}
Consider the first term on the right hand side (RHS) of~\eqref{eq:d1}. By~\eqref{eq:fofx}:
\begin{align}
\nonumber \sum_{x=S_{\min}}^{S_{\min}+m-1} \widetilde{f}(x) \ln \frac{\widetilde{f}(x)}{f(x)}&=\sum_{x=S_{\min}}^{S_{\min}+m-1} f(x)(1-p)  \ln(1-p),\\ 
\nonumber &=  (1-p) \ln(1-p) F (S_{\min}+m-1),\\
\nonumber  &\leq (1-p) \ln(1-p) \\
\label{eq:d11} &\leq (\xi-\xi p) \ln(\xi-\xi p),
\end{align}
\noindent where $F(\cdot)$ denotes the cumulative distribution function (CDF) of packet sizes, and the last step is true since $\xi\geq 1$ following from~\eqref{eq:xi}.

Consider the second term on the RHS of~\eqref{eq:d1} . Substituting $\widetilde{f}(x)$ from~\eqref{eq:fofx} yields:
\begin{align}
\nonumber &\sum_{x=S_{\min}+m}^{S_{\max}} \widetilde{f}(x)  \ln\frac{\widetilde{f}(x)}{f(x)}\\
\nonumber &= \sum_{x=S_{\min}+m}^{S_{\max}}  (f(x)+p f(x-m)) \ln\frac{f(x)+p f(x-m)}{f(x)},\\ 
\nonumber  &\stackrel{(a)}{\leq} \sum_{x=S_{\min}+m}^{S_{\max}}  f(x)(1+p \xi_0) \ln({1+p \xi_0}),\\ 
\nonumber  &\stackrel{(b)}{\leq} \sum_{x=S_{\min}+m}^{S_{\max}}  f(x)(\xi+p \xi) \ln({\xi+p \xi}),\\ 
\nonumber  &= ((\xi+p \xi) \ln({\xi+p \xi})) \sum_{x=S_{\min}+m}^{S_{\max}}  f(x),\\ 
\nonumber  &= ((\xi+p \xi) \ln({\xi+p \xi})) (1-F(S_{\min}+m-1)),\\ 
\label{eq:d12}  &\leq  (\xi+\xi p) \ln(\xi+\xi p).
\end{align}
\noindent where $(a)$ follows from~\eqref{eq:xi0}, and $(b)$ follows from~\eqref{eq:xi}. By~\eqref{eq:d1},~\eqref{eq:d11}, and~\eqref{eq:d12},
\begin{align} \label{eq:7}
\nonumber \mathcal{D}( \widetilde{f}(x)|| f(x)) &\leq   (\xi+\xi p) \ln(\xi+\xi p) + (\xi-\xi p) \ln(\xi-\xi p),\\
\nonumber &\stackrel{(c)}{\leq} (\xi+\xi p)(\xi p) + (\xi+\xi p)(-\xi p)=2\xi^2 p^2\\
&= \frac{2 \epsilon^2}{n} 
\end{align}
\noindent where $(c)$ is true since $\ln(x)\leq x-1$  for $x>0$, and the last step follows from substituting the value of $p$ given in~\eqref{eq:p}. By~\eqref{eq:5},~\eqref{eq:6},~\eqref{eq:7} $\mathbb{P}_{e} \geq \frac{1}{2}- \epsilon$, and thus Alice's insertion is covert.

 {\em (Number of bits)} The total number of inserted bits in the stream of $n$ packets transmitted by Jack is
 \begin{align}
 \nonumber  n_c=\sum_{i=1}^{ n } n_c(i),
 \end{align}
 where $n_c(i)$ is the number of bits inserted in the $i^{\mathrm{th}}$ packet form Jack. Recall that the key is generated independent of the packet stream. Therefore, Alice inserts $m$ bits in a packet if two independent events occur, the key selects the packet (with probability $ p=\frac{\epsilon}{\xi \sqrt{n}}$) and the size of the packet satisfies $S\leq S_{\max}-m$ (with probability $F(S_{\max}-m)$). Let $L=F(S_{\max}-m)$. Observe:
 \begin{align}
 \nonumber \mathbb{E}[n_c(i)] &= L p m. 
 \end{align}
 Note that $n_c(i)$s are i.i.d. since the locations of ones in the key are i.i.d. The law of large numbers yields:  
 \begin{align}
 \nonumber \lim\limits_{n \to \infty} \mathbb{P}\left(n_c> \frac{L m  pn}{2}\right)&= \lim\limits_{n \to \infty}\mathbb{P}\left(\frac{1}{n}\sum_{i=1}^{n} n_c(i)> \frac{L m p }{2}\right)\\ 
 &=1
 \end{align} 
Since $ p=\frac{\epsilon}{\xi \sqrt{n}}$ (given in~\eqref{eq:p}), 
\begin{align}
\mathbb{P}\left(n_c> \frac{L m   \epsilon \sqrt{n}}{2 \xi}\right) \to 1 \text{ as } n \to \infty. 
\end{align}
Hence, Alice can insert $\mathcal{O}(\sqrt{n})$ bits in a sequence of $n$ packets from Jack. 
 
{\it (Converse)} Willie uses a detector that is sufficient to limit Alice's bit insertion in the packets across all potential schemes. Suppose that Willie observes a packet sequence of length $n$ and wishes to detect whether Alice inserts information in the packets or not. Since he knows that packet sizes are i.i.d. with pmf $f(x)$, he knows the expected size of packets. Recall that the mean and variance of packet sizes are $\mu< \infty$ and $\sigma^2<\infty$, respectively. Hence, he calculates the average size of packets $U$ and performs a hypothesis test by setting a threshold $t$:
\begin{align}
U\underset{H_0}{\overset{H_1}{\gtrless}} \mu + t.
\end{align}
Consider $\mathbb{P}_{FA}$
\begin{align}
\nonumber \mathbb{P}_{FA}=\mathbb{P}(U>\mu+t|H_0) &= \mathbb{P}(U-\mu>t|H_0) ,\\
\nonumber &\leq \mathbb{P}(|U-\mu|>t|H_0) \leq \frac{\sigma^2}{n t^2}. 
\end{align}
\noindent where the last step follows from the Chebyshev’s inequality. Since $\sigma^2< \infty$, if Willie sets $t=\sqrt{\frac{\sigma^2}{\alpha n}}$, he can achieve arbitrary small probability of false alarm, i.e., 
\begin{align}
\label{eq:fa1}\mathbb{P}_{FA}  \leq \alpha,
\end{align}
for arbitrary $\alpha$. 

Next, we will show that if Alice inserts total number of $\omega(\sqrt{n})$ bits, Willie can achieve arbitrary small probability of missed detection as well, i.e., $\mathbb{P}_{MD}\leq \beta$ for arbitrary $\beta>0$. Assume Alice inserts $n_c$ bits of information. Therefore, 
\begin{align}
\nonumber \mathbb{P}_{MD}&=\mathbb{P}(U \leq \mu+t|H_1),\\
\nonumber &= \mathbb{P}\left(\frac{1}{n}\sum_{i=1}^{n} S_i + \frac{n_c}{n}\leq \mu+t\right),\\
\nonumber &= \mathbb{P}\left(\frac{1}{n}\sum_{i=1}^{n} S_i -\mu \leq t-\frac{n_c}{n}\right).
\end{align}
\noindent Since $t=\sqrt{\frac{\sigma^2}{\alpha n}}$, if $n_c=\omega({\sqrt{n}})$, for large enough $n$, $t-\frac{n_c}{n}<0$, and thus the WLLN yields 
\begin{align}
\nonumber \lim\limits_{n \to \infty}\mathbb{P}_{MD}& = 0.
\end{align}
\noindent Therefore, if Alice inserts total number of $\omega(\sqrt{n})$ bits, Willie can achieve $\mathbb{P}_{MD} < \beta$ for any $\beta>0$. Combined with~\eqref{eq:fa1}, if Alice inserts $\omega\left(\sqrt{n}\right)$ bits, Willie can choose a $t=\sqrt{\frac{\sigma^2}{ n \alpha}}$ to achieve any (small) $\alpha>0$ and $\beta>0$ desired.
\end{proof}

\section{Assumption 2: Packet sizes are i.i.d. with a pmf whose support is general} \label{sec:iid2}

In this section, we assume that packet sizes are i.i.d. with pmf $f(x)$ whose support is $S_1 < S_2<\ldots<S_K$, with $K\geq 2$. We determine the total number of bits that Alice can insert in the packets.   

\begin{thm} \label{thm:iid2} 
	If packet sizes are i.i.d. with pmf $f(x)$ that allows $K\geq 2$ possible sizes, Alice can covertly insert $\mathcal{O}(\sqrt{n})$ bits in a sequence of $n$ packets transmitted by Jack. Conversely, if Alice inserts $\omega\left(\sqrt{n}\right)$ bits in a sequence of $n$ packets, Willie can detect her with arbitrarily small error probability $\mathbb{P}_{e}$.
\end{thm}

\begin{proof} The construction and analysis follows from those of Theorem~\ref{thm:iid} with minor modifications. Alice generates a secret key of length $\mathcal{O} (\sqrt{n} \log{n})$ and shares it with Bob. According to the key, Alice selects packets for insertion. Let $S_i\in \{S_1,\ldots,S_K\}$ be the size of the packet selected by Alice according to the key. Since the proof for odd $K$ follows from that of even case with disregarding one possible packet size, here we assume $K$ is even. If $i\leq K/2$, then Alice adds $S_{i+ K/2 }- S_i$ bits to the end of the payload of the packet so that the size of the packet will be $S_{i+ K/2 }$, and she sets the flag bit to one; else, she only changes the flag bit to zero. Note that in this case, since $S_i-S_{i'<i}\geq 1$, Alice adds at least $ K/2 $ bits to each packets which is lower bounded by that of Theorem~\ref{thm:iid}. Bob employs the key and the flag bits to extract and remove Alice's bits from the packets. 
\end{proof}

\section{Assumption 3: Dependent Packet Sizes} \label{sec:niid}

In this section, we assume that packet sizes may be dependent. We determine the total number of bits that Alice can insert in the packets.   

\begin{thm} \label{thm:niid} 
	If the packet sizes are dependent, Alice can covertly insert on average $\mathcal{O}(c(n)/\sqrt{n})$ bits in a sequence of $n$ packets transmitted by Jack, where $c(n)$ is the average number of conditional pmfs of packet sizes, $f(x_i|x_{i-1},\ldots,x_1)$s, that have a support of minimum size two, i.e.,
\begin{align}
\label{eq:c}	c(n)=\sum_{i=1}^{n} \mathbb{E} \left[\mathbb{I}_{\{K_i\geq 2\}}\right],
	\end{align}
where $\mathbb{I}$ denotes the indicator function, and $K_i$ is the size of the support of $f(x_i|x_{i-1},x_{i-2},\ldots,x_1)$.
\end{thm}
 
\begin{proof}
	
	{\it (Achievability)} 
	
	\textbf{Construction}: The construction is similar to that of Theorem~\ref{thm:iid} except that the number of bits that Alice adds to the $i^{\mathrm{th}}$ packet from Jack depends on $K_i$. If $K_i$ is odd, Alice and Bob disregard packets with the smallest possible size. Alice generates a secret key of length $\mathcal{O}( \sqrt{n} \log{n})$ indicating the packets in the stream of length $n$ transmitted by Jack have to be selected by Alice for a possible bit insertion. She shares the key with Bob; however, Willie does not have access to it. To generate the key, she first generates an i.i.d. bit sequence of length $n$ in which each bit is one with probability 
\begin{align}
\label{eq:p2}  p'=\frac{\epsilon}{\eta \sqrt{n}},
\end{align}
where
\begin{align}
\label{eq:eta} \eta&=\max\{1,\eta_0\},\\
\label{eq:eta0} \eta_0&=\max\limits_{1 \leq i\leq n}{\frac{\max \mathcal{A}_i}{\min \mathcal{A}_i}},
\end{align}
$\mathcal{A}_i$ is the set of all possible sizes for the $i^{\mathrm{th}}$ for all possible  instantiations of sizes of $1^{\mathrm{st}},2^{\mathrm{nd}},\ldots,(i-1)^{\mathrm{th}}$ packets. The locations of the ones indicate the packets that will be selected by Alice for a possible bit insertion.	

If the key indicates that the $i^{\mathrm{th}}$ packet has to be selected by Alice for a possible bit insertion, $S_{1,i}<\ldots<S_{K_i,i}$ is the support of $f(x_i|x_{i-1},x_{i-2},\ldots,x_1)$, and $S_{j,i}$ is the size of the $i^{\mathrm{th}}$ packet from Jack, then Alice adds $S_{j+\lfloor K_i/2 \rfloor,i}- S_{j,i}$ bits to the end of the payload of the packet so that the size of the packet will be $S_{j+\lfloor K_i/2 \rfloor,i}$, and she sets the flag bit to one; else, she only changes the flag bit to zero.

Bob, who has access to the key and the flag bits finds out which packets contain Alice's bits. Therefore, he extracts and removes the last $m$ bits of those packets.  
	
	\textbf{Analysis}: {\em (Covertness)} Observe:
	\begin{align}
	\mathbb{P}_0&= \prod_{i=1}^{n}f(x_i|x_{i-1},\ldots,x_1),\\
	\mathbb{P}_1&= \prod_{i=1}^{n}\widetilde{f}(x_i|x_{i-1},\ldots,x_1).
	\end{align}
	\noindent From the chain rule for relative entropy\cite[Eq.(2.67)]{cover2012elements}:
	\begin{equation}\label{eq:61}
	\mathcal{D}(\mathbb{P}_1 || \mathbb{P}_0) = \sum_{i=1}^{n} \mathcal{D}(\widetilde{f}(x_i|x_{i-1},\ldots,x_1) || f(x_i|x_{i-1},\ldots,x_1)).
	\end{equation}
	Similar to the arguments leading to~\eqref{eq:fofx}, we can show that if $K_i\geq 2$, then \eqref{eq:p2} yields:
	\begin{align} \label{eq:72}
	 \mathcal{D}(\widetilde{f}(x_i|x_{i-1},\ldots,x_1) || f(x_i|x_{i-1},\ldots,x_1)) &\leq \frac{2 \epsilon^2}{n}. 
	\end{align}
	\noindent By~\eqref{eq:5},~\eqref{eq:61},~\eqref{eq:72}, $\mathbb{P}_{e} \geq \frac{1}{2}- \epsilon$, and thus Alice's bit insertion is covert. 
	
	{\em (Number of bits)} Recall that $K_i$ is size of the support of $f(x_i|x_{i-1},\ldots,x_1)$, and $W(i)=1$ if the key selects the $i^{\mathrm{th}}$ packet for a possible insertion by Alice. Alice inserts bits in the $i^{\mathrm{th}}$ packet from Jack if two independent events occur, $W(i)=1$ and $K_i\geq 1$. Thus, the total number of bits that Alice inserts is
	\begin{align}
	\nonumber  n_c=\sum_{i=1}^{n} \lfloor K_i/2\rfloor\mathbb{I}_{\{K_i\geq 2\}}  \mathbb{I}_{\{W(i)=1\}}&\geq \sum_{i=1}^{n} \mathbb{I}_{\{K_i\geq 2\}}  \mathbb{I}_{\{W(i)=1\}}
	\end{align}
	where the last step is true since if $K_i\geq 2$, then $\lfloor K_i/2\rfloor \geq 1$. In oter words, if Alice inserts bits in a packet, she inserts at least one bit in it. Consequently, 
	\begin{align}
	\nonumber \mathbb{E} [n_c]&\geq \sum_{i=1}^{n} \mathbb{E} \left[\mathbb{I}_{\{K_i\geq 2\}}  \mathbb{I}_{\{W(i)=1\}}\right],\\
	&\stackrel{(d)}{=} \sum_{i=1}^{n} \mathbb{E} \left[\mathbb{I}_{\{K_i\geq 2\}}\right]  p= p c(n),
	\end{align}
	\noindent where $(d)$ is true since $K_i$ is independent of $W(i)$ and $\mathbb{E}[\mathbb{I}_{\{W(i)=1\}}]=p'$. Because $p'=\frac{\epsilon}{\eta \sqrt{n}}$, Alice can insert on average $\mathcal{O}(c(n)/\sqrt{n})$ bits in a sequence of $n$ packets from Jack. 
\end{proof}

%%%%%%%%%%%%%%%%%%%%%%%%%%%%%
\section{Discussion}\label{dis}
\subsection{Use of Other Techniques for Covert Communication}
In this paper, we allowed Alice to hold packets to add additional bits of information to their payloads. This may allow her to release the packets at specific times to embed information in the inter-packet delays and establish a timing channel which results in a higher throughput~\cite{soltani2015covert,soltani2016allerton,soltani2017towards}. However, this paper only focuses on establishing a covert communication via insertion of bits in packets, and everything else that Alice does besides this, such as timing channel and packet insertion~\cite{soltani2015covert,soltani2016allerton} is orthogonal to this work. Furthermore, Alice cannot alter the information in Jack's packets to embed her information since if she does such, Steve may realize this and punish her. 

\subsection{Assumption of variable packet lengths} \label{dis2}
Our analysis and results relied on the assumption that the packet sizes are variable so Alice can insert her information in the payloads of the packets and increase the packet sizes. Although in some scenarios the packet sizes are almost deterministic~\cite[Section III.A]{zhang2014network}, there are many scenarios with variable packet lengths~\cite{garsva2015packet,farber2002network,salvador2004modeling} where a covert channel may be established by embedding information in the packet sizes~\cite{zhang2014network,yao2008coverting,ji2009novel,nair2011length,ji2009normal}. Besides, audio and video streaming applications with variable bit rate codecs (e.g., Skype) also might involve the transmission of variable packet sizes.  

\subsection{Recovering the original size of the packets}
In this paper, the mapping between the old sizes of the packets and new sizes of the packets is such that Bob can easily find the original size of each packet if it has information from Alice. Since he also knows that the bits are added to the end of the payload, he can extract Alice's bits. Considering Assumption 1, we discuss an alternative mapping: if the key selects a packet for a possible bit insertion by Alice, then Alice uses all of the available space in the packet. Since the mapping does not preserve the original sizes of the manipulated packets, Alice has to allocate some of the added bits to indicate the original size of the packets. We can show that we achieve similar order results for Theorems~\ref{thm:iid},\ref{thm:iid2}, and~\ref{thm:niid}.

\subsection{Alternative uses of key}
We assumed that Alice and Bob share a secret key that is unknown to Willie, and that Alice can use a flag bit in the header of each packet to indicate if she inserted any bits in a packet that is selected by the key. Here we discuss two alternative schemes. In the first scheme, Alice generates a key of size $\mathcal{O}(1)$ and shares it with Bob. If she inserts bits in a packet, she inserts the key in the packet. Bob finds the packets in which the key exists and extracts Alice's information. Because the key can appear in the payload randomly without Alice inserting it, one has to analyze the probability of such an error for Bob.

In the second scheme, instead of using a fixed key for each packet, Alice employs non-deterministic encryption to create a long secret key, then she slices it into pieces and inserts each piece in each packet that has Alice's bits. In this case, the size of the key grows as $n$ grows and it is efficient only in a model where the size of the packets scales as some function of $n$.

\subsection{Covertness of flag bits}
Since Willie cannot see the content of the packets, he cannot see the flag bits. We consider Assumption 1 and we show that even if he observes the flag bits, Alice's bit insertion is covert. We assume that if Alice does not insert bits, the flag bits are i.i.d. instantiations of a Bernoulli random variable with parameter $1/2$. Let $\mathbb{Q}_1$ and $\mathbb{Q}_0$ be the joint pmf of packet sizes and flag bits under $H_1$ and $H_0$, respectively. If Willie applies hypothesis test, then \cite{bash_jsac2013}
\begin{align} \label{eq:9} \mathbb{P}_{e} \geq \frac{1}{2}- \sqrt{\frac{1}{2} \mathcal{D}(\mathbb{Q}_1 || \mathbb{Q}_0)},   
\end{align}
where 
\begin{align}
\nonumber \mathbb{Q}_0&= \prod_{i=1}^{n}h(x_i,y_i),\\
\nonumber \mathbb{Q}_1&= \prod_{i=1}^{n}\widetilde{h}(x_i,y_i),
\end{align}
and $h(x,y)$ and $\tilde{h}(x,y)$ are the join pmf of packet size and flag bit under $H_1$ and $H_0$, respectively. From the chain rule for relative entropy \cite[Eq.(2.67)]{cover2012elements}:
\begin{align}\label{eq:10}
\mathcal{D}(\mathbb{Q}_1 || \mathbb{Q}_0) &= n \mathcal{D}(\widetilde{h}(x,y) || h(x,y)).
\end{align}
Recall that the pmf of packet sizes under $H_1$ and $H_0$ are $\tilde{f}(x)$ and $f(x)$, respectively. Since the size of a packet is independent of its flag bit when $H_0$ is true, $h(x,y)=f(x)B_{0.5}(y)$, where $B_{0.5}(y)$ is a Bernoulli pmf of parameter $1/2$. Let $\tilde{h}_{Y|X}(x,y)$ be the conditional pmf of the flag bits given the size of the packet under $H_1$. The chain rule of relative entropy yields $\mathcal{D}(\widetilde{h}(x,y) || h(x,y))= \mathcal{D}(\widetilde{f}(x) || f(x)) + \mathcal{D}(\widetilde{h}_{Y|X}(y,x) || B_{0.5}(y))$. Combined with~\eqref{eq:10}:
\begin{align}\label{eq:11}
\mathcal{D}(\mathbb{Q}_1 || \mathbb{Q}_0) &= n \mathcal{D}(\widetilde{f}(x) || f(x))+ n \mathcal{D}(\tilde{h}_{Y|X}(y,x) || B_{0.5}(y)).
\end{align}
Consider the first term on the RHS of~\eqref{eq:11}. Similar to the arguments leading to~\eqref{eq:7}, we can show that $ \mathcal{D}( \widetilde{f}(x)|| f(x)) \leq \frac{ \epsilon^2}{n}$. Thus, 
\begin{align} \label{eq:12}
n \mathcal{D}( \widetilde{f}(x)|| f(x)) &\leq  \epsilon^2.
\end{align}
\noindent Next, we derive $\tilde{h}_{Y|X}(y,x)$ and upper bound the RHS of~\eqref{eq:11}. Similar to the analysis of Theorem~\ref{thm:iid}, we assume that the number of possible sizes for each packet ($K$) is odd, as the analysis for even $K$ follows from it readily. By~\eqref{eq:fofx} and the Bayes' rule, we can show that for $S_{\min}\leq x \leq S_{\min}+m-1$, 
\begin{align}
\label{eq:13} \tilde{h}_{Y|X}(y=0,x) = \mathbb{P}(y=0|x)=\frac{0.5 f(x)(1-p)}{ f(x)(1-p)} = 0.5,
\end{align} 
and for $S_{\min}+m\leq x \leq S_{\max}$,
\begin{align}
\label{eq:14}
\tilde{h}_{Y|X}(y=0,x) = \mathbb{P}(y=0|x)=\frac{0.5 f(x)(1-p)}{ f(x)(1-p)+f(x-m)p}.
\end{align}
Consider the second term on the RHS of~\eqref{eq:11}. Observe: 
\begin{align}
\nonumber  &\mathcal{D}(\tilde{h}_{Y|X}(y,x) || B_{0.5}(y))= 
\\
\label{eq:18}&\sum_{x=S_{\min}+m}^{S_{\max}} \tilde{f}(x) \sum_{y=0}^{1} \tilde{h}_{Y|X}(y,x) \log{\frac{\tilde{h}_{Y|X}(y,x)}{2}}.
\end{align}
\noindent Note that the second summation on the RHS of~\eqref{eq:18} is only over $S_{\min}+m\leq x\leq x=S_{\max}$ because for $S_{\min}\leq x \leq S_{\min}+m-1$,~\eqref{eq:13} implies that $\tilde{h}_{Y|X}(y,x)$ is $B_{0.5}(x)$. We show in Appendix that
\begin{align}
\label{eq:16}\sum_{y=0}^{1} \tilde{h}_{Y|X}(y,x) \log{\frac{\tilde{h}_{Y|X}(y,x)}{2}}
\leq \left(\frac{ \xi_0 p}{ 1-p}\right)^2.
\end{align} 
Consequently,~\eqref{eq:18} yields:
\begin{align}
\nonumber  \mathcal{D}(\tilde{h}_{Y|X}(y,x) || B_{0.5}(y))&\leq  \left(\frac{ \xi_0 p}{ 1-p}\right)^2 \sum_{x=S_{\min}+m}^{S_{\max}} \tilde{f}(x),\\ 
\label{eq:19} &\leq \left(\frac{ \xi_0 p}{ 1-p}\right)^2,
\end{align}
where the last step is true since the summation in~\eqref{eq:19} is the probability that the packet size satisfies $S_{\min}+m\leq x \leq S_{\max}$ under $H_1$. Substituting the value of $p$ from~\eqref{eq:p} and~\eqref{eq:19} yield: 
\begin{align}
\label{eq:20} \lim_{n \to \infty} n  \mathcal{D}(\tilde{h}_{Y|X}(y,x) || B_{0.5}(y)) \leq   { \xi_0^2 \frac{\epsilon^2}{\xi^2}} \leq \epsilon^2, 
\end{align}
where the last step follows from~\eqref{eq:xi}. By~\eqref{eq:9},~\eqref{eq:11},~\eqref{eq:12}, and~\eqref{eq:20}, $\mathbb{P}_{e} \geq \frac{1}{2}-\epsilon$ for all $\epsilon$ as $n \to \infty$, and thus Alice's insertion is covert, even if Willie observes the flag bits.  

\subsection{Avoiding the use of flag bits}
The use of the flag bit~\cite{fisk2002eliminating,ahsan2002practical,shah2011network,cauich2005data,rowland1997covert,handel1996hiding}) are possible in various protocols where random-looking fields in the headers are generated (e.g., initial sequence numbers (ISN) in TCP, and IPID in IP). Considering Assumption 1, we discuss two alternative scheme where Alice and Bob do not use the flag bits. 
In the first scheme, if the key selects a packet, and the size of the selected packet satisfies $S\leq S_{\max}-m$, then Alice inserts $m$ bits to the end of the payload of the packet; otherwise, she removes and stores $S_{\max}-m$ bits from the payload of the packet, and for bit insertion, she gives the highest priority to the oldest Jack's bits in her buffer. She continues this until she transmits Jack's last packet. Note that the mapping is one to one and thus Bob can find the original size of each packet and extract Alice's bits. We can show that if the pmf of the packet sizes satisfies some conditions, with high probability Alice can transmit all Jack's bits from her buffer as well as $\mathcal{O}(\sqrt{n})$ bits of her own. The packet stream transmitted by Bob has all Jack's bits, but if Jack transmits the $k^{\mathrm{th}}$ bit in the $l^{\mathrm{th}}$ packet, the $k^{\mathrm{th}}$ bit is not necessarily in the $l^{\mathrm{th}}$ packet transmitted by Bob. 

In the second scheme, Alice uses dummy bits. Assume Alice observes a packet that is selected by the key and that the size of the packet is $S$. Alice chooses an integer $\ell$, where $\ell < \lfloor K/2 \rfloor $. If $S \leq S_{\max}-2\ell$, Alice adds $\ell$ of her bits to the end of the packet. If $S > S_{\max}-2\ell$, she adds $\ell$ dummy packets. Assume Bob observes a packet that is selected by the key and the size of the packet is $S'$. If $S'\leq S_{\max}-\ell$, he extracts $\ell$ bits from the end of the packet. The reason that Alice has to insert dummy bits in the packets of size $S > S_{\max}-2\ell$ is that the new size will be $S'>S_{\max}-\ell$ so Bob will have uncertainty about whether they contain $\ell$ bits from Alice or since they did not have enough space Alice did not insert any bits in them. Note that Bob cannot remove the dummy bits so that Steve will suffer from dummy bits in some packets of size $S > S_{\max}-\ell+1$. Alternatively, Alice can add $S_{\max}-1-S$ of her won bits to a packet of size $S<S_{\max}-1$, and one dummy bit to a packet of size $S_{\max}-1$. In this case, only some of the packets of maximum size will have dummy bits; however, since the mapping between the old and new size does not preserve the old size, some of the inserted bits have to be used to indicate the original size of the packets. 

\section{Future Work} \label{sec:f}
In the future work, we will assume that Willie can see the contents of the packets so he can employ entropy attacks to detects Alice's bit insertion, i.e., if he knows the $i^{\mathrm{th}}$ order of the entropy of Jacks' bits, he can verify that to restrict Alice's bit insertion. Therefore, Alice has to employ a coding where generates bits whose entropy are close Jack's bits. Even if Alice does not employ such this scheme, since Alice only manipulates $\mathcal{O}(\sqrt{n})$ packets, the conjecture is that the change of the entropy of the bits is covert and Alice can still insert $\mathcal{O}(\sqrt{n})$ bits when the packet sizes are i.i.d. We will verify this conjecture in the future work. Furthermore, we will consider the compression of Jacks bits to achieve more space available in the packets for Alice's bit insertion. 

\section{Conclusion} \label{con}
In this paper, we presented the fundamental limits of a third technique for the transmission of covert information on network links, covert insertion of bits in the payload of packets with available space, where the first and second schemes were packet insertion and timing channel presented in~\cite{soltani2015covert,soltani2016allerton}. In a network link where an authorized transmitter Jack sends a flow of $n$ packets intended for Steve, and the flow visits Alice (covert transmitter), Willie (network warden), Bob (covert receiver), and Steve (intended receiver), respectively, we have established that if the packet sizes are i.i.d. with a probability mass function whose support allows at least two possible packet sizes, Alice can insert $\mathcal{O}(\sqrt{n})$ bits in the packets of the flow while lower bounding Willie's error probability $\mathbb{P}_{e}$ by $\frac{1}{2}-\epsilon$ for all $\epsilon>0$. Furthermore, we showed that if Alice inserts $\omega(\sqrt{n})$ bits, Willie will detect her with an arbitrarily small error probability $\mathbb{P}_{e}$. For the case of dependent packet sizes, we showed that the average number of bits that Alice can insert in a stream of $n$ packets from Jack is $\mathcal{O}(c(n)/\sqrt{n})$, where $c(n)$ is the average number of conditional pmfs of packet sizes given the history with support of minimum size two.

\appendix

%\label{ap1} \subsection{proof of~\eqref{eq:16}}
By~\eqref{eq:14}, 
\begin{align}
\nonumber &\sum_{y=0}^{1} \tilde{h}_{Y|X}(y,x) \log{\frac{\tilde{h}_{Y|X}(y,x)}{2}}\\  
\label{eq:17}   &= \frac{1-r}{2} \log (1-r) + \frac{1+r}{2} \log(1+r)\leq r^2,
\end{align}
where $r=\frac{ f(x-m)p}{ f(x)(1-p)+f(x-m)p}$, and the last step is true since $0\leq r\leq 1$, $\log(1+r)\leq r$, and $\log(1-r)\leq -r$. Note that $r^2\leq \left(\frac{ f(x-m)p}{ f(x)(1-p)}\right)^2$. By~\eqref{eq:xi0}, $r^2\leq \left(\frac{ \xi_0 p}{ 1-p}\right)^2$. Combined with~\eqref{eq:17}, the proof is complete. 

\bibliographystyle{ieeetr}

%\bibliography{mymaincitation}

\end{document}